\documentclass[3p,twocolumn]{elsarticle}
\usepackage[utf8x]{inputenc}
\usepackage{amsmath,amssymb}
\usepackage{tikz}
\usetikzlibrary{shapes.geometric,shapes.misc}
\usetikzlibrary{positioning,calc}
\usepackage{algorithmicx}
\usepackage[noend]{algpseudocode}
\usepackage{hyperref}
\usepackage{subcaption}

\newtheorem{proposition}{Proposition}
\newtheorem{definition}{Definition}
\newtheorem{theorem}{Theorem}
\newtheorem{lemma}{Lemma}
\newproof{proof}{Proof}

\def\NP{\mathbf{NP}}
\def\root{{\textsc{root}}}
\def\size{{\textsc{size}}}
\def\parent{{\textsc{parent}}}
\def\merge{{\textsc{merge}}}
\def\compress{\textsc{collapse}}
\def\children{{\textsc{children}}}
\def\push{{\textsc{push}}}
\def\sumsize{{\textsc{sumsize}}}
\def\sumlight{{\textsc{sumlight}}}
\def\levelone#1{\children(#1)}

\def\drawweight#1#2#3#4#5#6{
{
\node (r) at (#1,#2) [draw] {};
\path (r) edge #6;
\foreach \i in {1,...,#3}{
  \node (c\i) at ({#1-#4/2+(\i-1)*#4/(#3-1)},#5) [draw] {};
  \path (c\i) edge (r);
}
}
}

\begin{document}
\title{Recognizing Union-Find trees is $\NP$-complete\tnoteref{thanksfn}}
\author[kitti]{Kitti~Gelle}
\author[szabivan]{Szabolcs~Iv\'an}
\address[kitti,szabivan]{University of Szeged, Hungary}
\tnotetext[thanksfn]{This research was supported by NKFI grant number K108448.}
\begin{abstract}
Disjoint-Set forests, consisting of Union-Find trees are data structures having a widespread practical application due to their efficiency.
Despite them being well-known, no exact structural characterization of these trees is known (such a characterization exists for 
Union trees which are constructed without using path compression).
In this paper we provide such a characterization by means of a simple $\push$ operation
and show that the decision problem whether a given tree is a Union-Find tree is $\NP$-complete.
\end{abstract}
\maketitle
\section{Introduction}
Disjoint-Set forests, introduced in~\cite{Galler:1964:IEA:364099.364331},
are fundamental data structures in many practical algorithms where one has to maintain a partition of some set,
which support three operations: \emph{creating} a partition consisting of singletons, \emph{querying} whether two given elements are
in the same class of the partition (or equivalently: \emph{finding} a representative of a class, given an element of it)
and \emph{merging} two classes.
Practical examples include e.g.
building a minimum-cost spanning tree of a weighted graph~\cite{Cormen:2001:IA:580470},
unification algorithms~\cite{Knight:1989:UMS:62029.62030}
etc.

To support these operations, even a linked list representation suffices but to achieve an almost-constant amortized time cost per operation, Disjoint-Set forests are used in practice. In this data structure, sets are represented as directed trees with the edges directed towards the root; the $\textsc{create}$ operation creates $n$ trees having one node each (here $n$ stands for the number of the elements in the universe), the
$\textsc{find}$ operation takes a node and returns the root of the tree in which the
node is present (thus the $\textsc{same-class}(x,y)$ operation is implemented as
$\textsc{find}(x)==\textsc{find}(y)$), and the $\textsc{merge}(x,y)$ operation
is implemented by merging the trees containing $x$ and $y$, i.e. making one of the root nodes to be a child of the other root node (if the two nodes are in different classes).

In order to achieve near-constant efficiency, one has to keep the (average) height of the trees small.
There are two ``orthogonal'' methods to do that: first, during the merge operation it is advisable to attach the ``smaller'' tree below the ``larger''
one. If the ``size'' of a tree is the number of its nodes, we say the trees are built up according to the \emph{union-by-size} strategy,
if it's the depth of a tree, then we talk about the \emph{union-by-rank} strategy. Second, during a $\textsc{find}$ operation invoked on some node $x$
of a tree, one can apply the \emph{path compression} method, which reattaches each ancestor of $x$ directly to the root of the tree in which they are
present. If one applies both the path compression method and either one of the union-by-size or union-by-rank strategies, then any sequence of
$m$ operations on a universe of $n$ elements has worst-case time cost $O(m\alpha(n))$ where $\alpha$ is the inverse of the extremely fast growing
(not primitive recursive) Ackermann function for which $\alpha(n)\leq 5$ for each practical value of $n$ (say, below $2^{65535}$), hence it has
an amortized almost-constant time cost~\cite{Tarjan:1975:EGB:321879.321884}. Since it's proven~\cite{Fredman:1989:CPC:73007.73040} that \emph{any} data structure has worst-case time cost $\Omega(m\alpha(n))$, the Disjoint-Set forests equipped with a strategy and path compression offer a theoretically optimal data structure which performs exceptionally well also in practice.
For more details see standard textbooks on data structures, e.g.~\cite{Cormen:2001:IA:580470}.

Due to these facts, it is certainly interesting both from the theoretical as well as the practical point of view to characterize those trees
that can arise from a forest of singletons after a number of merge and find operations, which we call Union-Find trees in this paper.
One could e.g. test Disjoint-Set implementations since if at any given point of execution a tree of a Disjoint-Set forest is not a valid
Union-Find tree, then it is certain that there is a bug in the implementation of the data structure (though we note at this point that this
data structure is sometimes regarded as one of the ``primitive'' data structures, in the sense that is is possible to implement a correct
version of them that needs not be certifying~\cite{DBLP:journals/csr/McConnellMNS11}). Nevertheless, only the characterization
of Union trees is known up till now~\cite{DBLP:journals/ipl/Cai93}, i.e. which correspond to the case when one uses one of the union-by- strategies but
\emph{not} path compression. Since in that case the data structure offers only a theoretic bound of $\Theta(\log n)$ on the amortized time cost,
in practice all implementations imbue path compression as well, so for a characterization to be really useful, it has to cover this case
as well.

In this paper we show that the recognition problem of Union-Find trees is $\NP$-complete when the union-by-size strategy is used
(and leave open the case of the union-by-rank strategy).
This confirms the statement from~\cite{DBLP:journals/ipl/Cai93} that the problem ``seems to be much harder''
than recognizing Union trees (which in turn can be done in low-degree polynomial time).

{\bf Related work.} There is an increasing interest in determining the complexity of the recognition problem of various
data structures. The problem was considered for suffix trees~\cite{I2014316,Starikovskaya201514},
(parametrized) border arrays~\cite{I20116959,Lu2002,Duval:2005:BAB:1131983.1131987,MR2544434,MR2894365},
suffix arrays~\cite{Bannai2003208,Duval2002249,Kucherov2013915},
KMP tables~\cite{Duval2009281,Gawrychowski2014337},
prefix tables~\cite{DBLP:conf/stacs/ClementCR09}, 
cover arrays~\cite{Crochemore2010251}, and directed acyclic word- and subsequence graphs~\cite{Bannai2003208}.

\section{Notation}
A \emph{tree} is a tuple $t=(V_t,\root_t,\parent_t)$ with $V_t$ being the finite set of its \emph{nodes}, $\root_t\in V_t$ its \emph{root} and $\parent_t:(V_t-\{\root_t\})\to V_t$ mapping each non-root node to its
\emph{parent} (so that the graph of $\parent_t$ is a directed acyclic graph, with edges being directed towards the root).

For a tree $t$ and a node $x\in V_t$, let $\children(t,x)$ stand for the set
$\{y\in V_t:\parent_t(y)=x\}$ of its \emph{children} and $\levelone{t}$ stand as a shorthand
for $\children(t,\root_t)$, the set of depth-one nodes of $t$. Two nodes are \emph{siblings} in $t$ if they have the same parent.
Also, let $x\preceq_t y$ denote that $x$ is an \emph{ancestor} of $y$ in $t$, i.e. $x=\parent_t^k(y)$
for some $k\geq 0$ and let $\size(t,x)=|\{y\in V_t:x\preceq_t y\}|$ stand for the number
of \emph{descendants} of $x$ (including $x$ itself). Let $\size(t)$ stand for $\size(t,\root_t)$, the number of nodes in the tree $t$.
For $x\in V_t$, let $t|_x$ stand for the \emph{subtree} $(V_x=\{y\in V_t:x\preceq_t y\},x,\parent_t|_{V_x})$ of $t$ rooted at $x$.
When $x,y\in V_t$, we say that $x$ is \emph{lighter} than $y$ in $t$ (or $y$ is \emph{heavier} than $x$) if $\size(t,x)<\size(t,y)$.

Two operations on trees are that of \emph{merging} and \emph{collapsing}.
Given two trees $t=(V_t,\root_t,\parent_t)$ and $s=(V_s,\root_s,\parent_s)$ with $V_t$ and $V_s$ being disjoint,
their merge $\textsc{merge}(t,s)$ (in this order) is the tree $(V_t\cup V_s,\root_t,\parent)$ with
$\parent(x)=\parent_t(x)$ for $x\in V_t$, $\parent(\root_s)=\root_t$ and $\parent(y)=\parent_s(y)$ for each non-root node $y\in V_s$ of $s$.
Given a tree $t=(V,\root,\parent)$ and a node $x\in V$, then $\compress(t,x)$ is the tree $(V,\root,\parent')$
with $\parent'(y)=\root$ if $y$ is a non-root ancestor of $x$ in $t$, and $\parent'(y)=\parent(y)$ otherwise.
For examples, see Figure~\ref{fig-merge-collapse-push}.
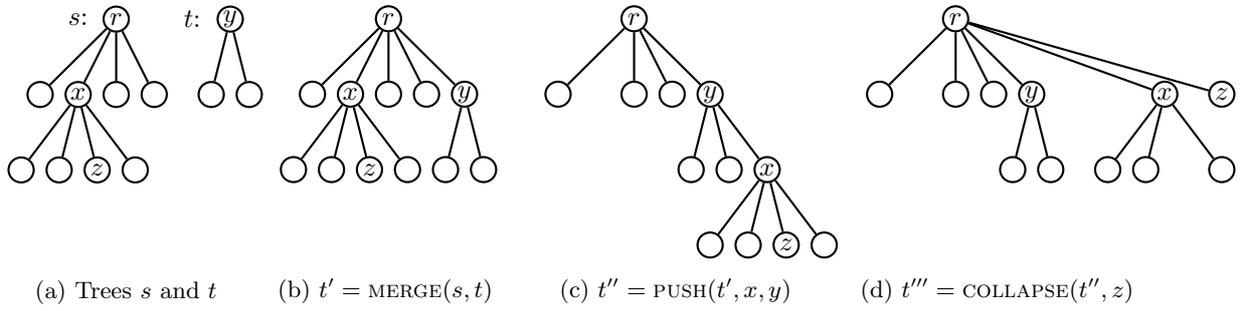
\begin{figure*}[h]
\begin{subfigure}[c]{0.2\textwidth}
\centering
\begin{tikzpicture}[thick]
\draw[white] (0,1.2) rectangle (1,-2.2);
\node at (1,1) {$s$:};
\node at (2.5,1) {$t$:};
\node[circle] (r) at (1.5,1) [draw,radius=0.2cm] {};
\foreach \i in {1,...,4}{
\node[circle] (c\i) at (\i/2,0) [draw,radius=0.2cm] {};
\path (r) edge (c\i);
}
\foreach \i in {0,...,3}{
\node[circle] (cc\i) at (\i/2+0.25,-1) [draw,radius=0.2cm] {};
}
\path (cc0) edge (c2);
\path (cc1) edge (c2);
\path (cc2) edge (c2);
\path (cc3) edge (c2);
\node[left of=c2, node distance=0mm] {$x$};
\node[left of=cc2, node distance=0mm] {$z$};
\node[left of=r, node distance=0mm] {$r$};
\node[circle] (r2) at (3,1) [draw,radius=0.2cm] {};
\foreach \i in {0,...,1}{
\node[circle] (c\i) at (\i/2+2.75,0) [draw,radius=0.2cm] {};
\path (r2) edge (c\i);
}
\node[left of=r2, node distance=0mm] {$y$};
\end{tikzpicture}
\caption{Trees $s$ and $t$}
\end{subfigure}
\begin{subfigure}[c]{0.20\textwidth}
\centering
\begin{tikzpicture}[thick]
\draw[white] (0,1.2) rectangle (1,-2.2);
\node[circle] (r) at (1.5,1) [draw,radius=0.2cm] {};
\foreach \i in {1,...,5}{
\node[circle] (c\i) at (\i/2,0) [draw,radius=0.2cm] {};
\path (r) edge (c\i);
}
\foreach \i in {0,...,5}{
\node[circle] (cc\i) at (\i/2+0.25,-1) [draw,radius=0.2cm] {};
}
\path (cc0) edge (c2);
\path (cc1) edge (c2);
\path (cc2) edge (c2);
\path (cc3) edge (c2);
\path (cc4) edge (c5);
\path (cc5) edge (c5);
\node[left of=c2, node distance=0mm] {$x$};
\node[left of=c5, node distance=0mm] {$y$};
\node[left of=cc2, node distance=0mm] {$z$};
\node[left of=r, node distance=0mm] {$r$};
\end{tikzpicture}
\caption{$t'=\merge(s,t)$}
\end{subfigure}
\begin{subfigure}[c]{0.25\textwidth}
\centering
\begin{tikzpicture}[thick]
\draw[white] (0,1.2) rectangle (1,-2.2);
\node[circle] (r) at (1.5,1) [draw,radius=0.2cm] {};
\foreach \i in {1,3,4,5}{
\node[circle] (c\i) at (\i/2,0) [draw,radius=0.2cm] {};
\path (r) edge (c\i);
}
\foreach \i in {4,5}{
\node[circle] (cc\i) at (\i/2+0.25,-1) [draw,radius=0.2cm] {};
}
\foreach \i in {2}{
\node[circle] (c\i) at (3.25,-1) [draw,radius=0.2cm] {};
}
\foreach \i in {0,...,3}{
\node[circle] (cc\i) at (\i/2+2.5,-2) [draw,radius=0.2cm] {};
}
\path (c2) edge (c5);
\path (cc0) edge (c2);
\path (cc1) edge (c2);
\path (cc2) edge (c2);
\path (cc3) edge (c2);
\path (cc4) edge (c5);
\path (cc5) edge (c5);
\node[left of=c2, node distance=0mm] {$x$};
\node[left of=c5, node distance=0mm] {$y$};
\node[left of=cc2, node distance=0mm] {$z$};
\node[left of=r, node distance=0mm] {$r$};
\end{tikzpicture}
\caption{$t''=\push(t',x,y)$}
\end{subfigure}
\begin{subfigure}[c]{0.25\textwidth}
\centering
\begin{tikzpicture}[thick]
\draw[white] (0,1.2) rectangle (1,-2.2);
\node[circle] (r) at (1.5,1) [draw,radius=0.2cm] {};
\foreach \i in {1,3,4,5}{
\node[circle] (c\i) at (\i/2,0) [draw,radius=0.2cm] {};
\path (r) edge (c\i);
}
\foreach \i in {4,5}{
\node[circle] (cc\i) at (\i/2+0.25,-1) [draw,radius=0.2cm] {};
}
\foreach \i in {2}{
\node[circle] (c\i) at (4.25,0) [draw,radius=0.2cm] {};
}
\foreach \i in {0,1,3}{
\node[circle] (cc\i) at (\i/2+3.5,-1) [draw,radius=0.2cm] {};
}
\foreach \i in {2}{
\node[circle] (cc\i) at (5,0) [draw,radius=0.2cm] {};
}
\path (c2) edge (r);
\path (cc0) edge (c2);
\path (cc1) edge (c2);
\path (cc2) edge (r);
\path (cc3) edge (c2);
\path (cc4) edge (c5);
\path (cc5) edge (c5);
\node[left of=c2, node distance=0mm] {$x$};
\node[left of=c5, node distance=0mm] {$y$};
\node[left of=cc2, node distance=0mm] {$z$};
\node[left of=r, node distance=0mm] {$r$};
\end{tikzpicture}
\caption{$t'''=\compress(t'',z)$}
\end{subfigure}
\caption{Merge, collapse and push.}
\label{fig-merge-collapse-push}
\end{figure*}

The class of \emph{Union trees} is the least class of trees satisfying the following two conditions:
every \emph{singleton tree} (having exactly one node) is a Union tree, and if $t$ and $s$ are Union trees with $\size(t)\geq\size(s)$,
then $\textsc{merge}(t,s)$ is a Union tree as well.

Analogously, the class of \emph{Union-Find trees} is the least class of trees satisfying the following three conditions:
every singleton tree is a Union-Find tree, if $t$ and $s$ are Union-Find trees with $\size(t)\geq\size(s)$, then
$\textsc{merge}(t,s)$ is a Union-Find tree as well, and if $t$ is a Union-Find tree and $x\in V_t$ is a node of $t$,
then $\compress(t,x)$ is also a Union-Find tree.

We'll frequently sum the size of ``small enough'' children of nodes, so we introduce a shorthand also for that: for a tree $t$, a node $x$ of $t$, and a threshold $W\geq 0$, let
$\sumsize(t,x,W)$ stand for $\sum\{\size(t,y):y\in\children(t,x),\size(t,y)\leq W\}$.
We say that a node $x$ of a tree $t$ \emph{satisfies the Union condition} if for each child $y$ of $x$ we have $\sumsize(t,x,W)\geq W$ where $W=\size(t,y)-1$. Otherwise, we say that $x$ \emph{violates the Union condition} (at child $y$).
Then, the characterization of Union trees from~\cite{DBLP:journals/ipl/Cai93} can be formulated in our terms as follows:
\begin{theorem}
\label{thm-union}
A tree $t$ is a Union tree if and only if each node $x$ of $t$ satisfies the Union condition.
\end{theorem}
Equivalently, $x$ satisisfies the Union condition if and only if whenever $x_1,\ldots,x_k$ is an enumeration of $\children(t,x)$ such that $\size(t,x_i)\leq\size(t,x_{i+1})$
for each $i$, then $\sum_{j<i}\size(t,x_j)\geq\size(t,x_i)-1$ for each $i=1,\ldots,k$.
(In particular, each non-leaf node has to have a leaf child.)

\section{Structural characterization of Union-Find trees}

Suppose $t$ and $s$ are trees on the same set $V$ of nodes and with the same root $r$.
We write $t\preceq s$ if $x\preceq_t y$ implies $x\preceq_s y$ for each $x,y\in V$.
For an example, consult Figure~\ref{fig-merge-collapse-push}. There, $t'''\preceq t''$
(e.g. $r\preceq_{t'''}x$ and also $r\preceq_{t''}x$);
the reverse direction does not hold since e.g. $y\preceq_{t''}z$ but
$y\not\preceq_{t'''}z$. Also, the reader is encouraged to verify that $t'''\preceq t'\preceq t''$ also holds.

Clearly, $\preceq$ is a partial order on any set of trees (i.e. is a reflexive, transitive and antisymmetric relation).
It is also clear that $t\preceq s$ if and only if $\parent_t(x)\preceq_s x$ holds for each $x\in V-\{r\}$
which is further equivalent to requiring $\parent_t(x)\preceq_s \parent_s(x)$ since $\parent_t(x)$ cannot be $x$.
%

Another notion we define is the (partial) operation $\push$ on trees as follows: when $t$ is a tree and $x\neq y\in V_t$ are siblings in $t$,
then $\push(t,x,y)$ is defined as the tree $(V_t,\root_t,\parent')$ with
$\parent'(z)=\begin{cases}y&\hbox{if }z=x,\\\parent_t(z)&\hbox{otherwise}\end{cases}$, that is, we ``push'' the node $x$
one level deeper in the tree just below its former sibling $y$.
(See Figure~\ref{fig-merge-collapse-push}.)  

We write $t\vdash t'$ when $t'=\push(t,x,y)$ for some $x$ and $y$, and as usual,
$\vdash^*$ denotes the reflexive-transitive closure of $\vdash$.
Observe that when $t'=\push(t,x,y)$, then $\size(t',y)=\size(t,y)+\size(t,x)>\size(t,y)$
and $\size(t',z)=\size(t,z)$ for each $z\neq y$, hence $\vdash^*$ is also a partial ordering on trees.
This is not a mere coincidence:
\begin{proposition}
\label{prop-push-is-lift}
For any pair $s$ and $t$ of trees, $t\preceq s$ if and only if $t\vdash^*s$.
\end{proposition}
\begin{proof}	
For $\vdash^*$ implying $\preceq$ it suffices to show that $\vdash$ implies $\preceq$ since the latter is a partial order.
So let $t=(V,r,\parent)$, and $x\neq y\in V$ be siblings in $t$  with the common parent $z$, and $s=\push(t,x,y)$.
Then, since $\parent_t(x)=z=\parent_s(y)=\parent_s(\parent_s(x))$, we get $\parent_t(x)\preceq_s x$, and by
$\parent_t(w)=\parent_s(w)$ for each node $w\neq x$, we have $t\preceq s$.

It is clear that $\preceq$ is equality on singleton trees, thus $\preceq$ implies $\vdash^*$ for trees of size $1$.
We apply induction on the size of $t=(V,r,\parent)$ to show that whenever $t\preceq s=(V,r,\parent')$, then
$t\vdash^* s$ as well. Let $X$ stand for the set $\children(t)$ of depth-one nodes of $t$ and $Y$ stand for $\children(s)$.
Clearly, $Y\subseteq X$ since by $t\preceq s$, any node $x$ of $t$ having depth at least two has to satisfy $\parent(x)\preceq_s\parent'(x)$
and since $\parent(x)\neq r$ for such nodes, $x$ has to have depth at least two in $s$ as well.
Let $\{x_1,\ldots,x_k\}$ stand for $X-Y$.
Now for any node $x_i\in X-Y$ there exists a unique node $y_i\in Y$ such that $y_i\preceq_s x_i$.
Let us define the trees $t_0=t$, $t_i=\push(t_{i-1},x_i,y_i)$.
Then $t\vdash^*t_k$, $\children(t_k)=Y=\children(s)$ and for each $y\in Y$, $t_k|_y\preceq s|_y$. Applying the induction hypothesis
we get that $t_k|_y\vdash^*s|_y$ for each $y\in Y$, hence the immediate subtrees of $t_k$ can be transformed into the immediate subtrees of $s$
by repeatedly applying $\push$ operations, hence $t\vdash^* s$ as well.
\end{proof}

The relations $\preceq$ and $\vdash^*$ are introduced due to their intimate relation to Union-Find and Union trees:
\begin{theorem}
\label{thm-uf-push}
A tree $t$ is a Union-Find tree if and only if $t\vdash^* s$ for some Union tree $s$.
\end{theorem}
\begin{proof}

Let $t$ be a Union-Find tree. We show the claim by structural induction.
For singleton trees the claim holds since any singleton tree is a Union tree as well. 
Suppose $t=\merge(t_1,t_2)$. Then by the induction hypothesis, $t_1\vdash^* s_1$ and
$t_2\vdash^* s_2$ for the Union trees $s_1$ and $s_2$. Then, for the tree $s=\merge(s_1,s_2)$ we get that $t\vdash^* s$.
Finally, assume $t=\compress(t',x)$ for some node $x$. Let $x=x_1\succ x_2\succ\ldots\succ x_k=\root_{t'}$ be the ancestral sequence of $x$ in $t'$. Then, defining $t_0=t$, $t_i=\push(t_{i-1},x_i,x_{i+1})$ we get that $t\vdash^* t_{k-1}=t'$ and $t'\vdash^* s$ for some Union tree applying the induction hypothesis, thus $t\vdash^* s$ also holds.

Now assume $t\vdash^* s$ (equivalently, $t\preceq s$ by Proposition~\ref{prop-push-is-lift}) for some Union tree $s$.
Let $X$ stand for the set $\children(t)$ of depth-one nodes of $t$ and $Y$ stand for $\children(s)$.
By $t\preceq s$ we get that $Y\subseteq X$. Let $\{x_1,\ldots,x_k\}$ stand for the set $X-Y$.
Then for each $x_i$ there exists a unique $y_i\in Y$ with $y_i\preceq_s x_i$.
Let us define the sequence $t=t_0\vdash t_1\vdash\ldots\vdash t_k$ with $t_i=\push(t_{i-1},x_i,y_i)$. Then, $t_k\preceq s$ holds.
Moreover, as $\children(t_k)=\children(s)=Y$, we get that $t_k|_y\preceq s|_y$ for each $y\in Y$.
Applying the induction hypothesis we have that each subtree $t_k|_y$ is a Union-Find tree.
Now let us define the sequence $t_0',t_1',\ldots,t_k'$ of trees as follows:
$t_0'$ is the singleton tree with root $\root_t$, and $t_i'=\merge(t'_{i-1},t_k|_{y'_i})$
where $Y=\{y'_1,\ldots,y'_\ell\}$ is an enumeration of the members of $Y$ such that $\size(s,y'_i)\leq 1+\mathop\sum\limits_{j<i}\size(s,y'_j)$
for each $i=1,\ldots,\ell$. (Such an ordering exists as $s$ is a Union tree.)
Since $\size(t'_{i-1})=1+\mathop\sum\limits_{j<i}\size(s,y'_j)$, we get thet each $t'_i$ is a Union-Find tree as well.

Finally, the tree $t$ results from the tree $t'_k$ constructed above
by applying successively one $\compress$ operation on each node in $X-Y$, hence
$t$ is a Union-Find tree as well.

\end{proof}

\section{Complexity}
In this section we show that the recognition of Union-Find trees is $\mathbf{NP}$-complete.
Note that membership in $\mathbf{NP}$ is clear by Theorem~\ref{thm-uf-push} and that
the possible number of pushes is bounded above by $n^2$:
upon pushing $x$ below $y$, the size of $y$ increases, while the size of the
other nodes remains the same. Since the size of any node is at most $n$, the
sum of the sizes of all the nodes is at most $n^2$ in any tree.

Let $H>0$ be a parameter (which we call the ``heaviness threshold parameter'' later on)
to be specified later and $t=(V,\root,\parent)$ be a tree. We call
a node $x$ \emph{light} if $\size(t,x)\leq H$, \emph{heavy} if $\size(t,x)>H$, and a \emph{basket} if it has a heavy child in $t$.
Note that every basket is heavy. A particular type of basket is an \emph{empty basket} which is a node having a single child of size $H+1$.
(Thus empty baskets have size $H+2$.)

Let us call a tree $t_0$ \emph{flat} if it satisfies all the following conditions:
\begin{enumerate}
\item There are $K>0$ depth-one nodes of $t_0$ which are empty baskets.
\item There is exactly one non-basket heavy depth-one node of $t_0$, with size $H+1$.
\item The total size of light depth-one nodes is $(K+1)\cdot H$. That is, $\sumsize(t_0,\root_{t_0},H)=(K+1)\cdot H$.
\item The light nodes and non-basket heavy nodes have only direct children as descendants,
	thus in particular, the subtrees rooted at these nodes are Union trees.
\end{enumerate}
See Figure~\ref{fig-flat-tree}.

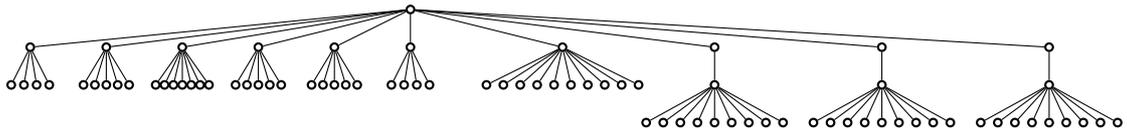
\begin{figure*}[h!]
\centering
\begin{tikzpicture}[every node/.style={thick,circle,inner sep=0pt,minimum size=0.1cm}]
\node (root) at (5,0) [draw] {};
\drawweight{0}{-0.5}{4}{0.5}{-1}{(root)}
\drawweight{1}{-0.5}{5}{0.6}{-1}{(root)}
\drawweight{2}{-0.5}{7}{0.7}{-1}{(root)}
\drawweight{3}{-0.5}{5}{0.6}{-1}{(root)}
\drawweight{4}{-0.5}{5}{0.6}{-1}{(root)}
\drawweight{5}{-0.5}{4}{0.5}{-1}{(root)}
\drawweight{7}{-0.5}{10}{2}{-1}{(root)}
{\node (r1) at (9,-0.5) [draw] {}; \path (r1) edge (root); \drawweight{9}{-1}{9}{1.8}{-1.5}{(r1)} }
{\node (r1) at (11.2,-0.5) [draw] {}; \path (r1) edge (root); \drawweight{11.2}{-1}{9}{1.8}{-1.5}{(r1)} }
{\node (r1) at (13.4,-0.5) [draw] {}; \path (r1) edge (root); \drawweight{13.4}{-1}{9}{1.8}{-1.5}{(r1)} }
\end{tikzpicture}
\caption{A flat tree with $H=9$, $K=3$. The six depth-one light nodes' sizes sum up to $4\times 9=36$.}
\label{fig-flat-tree}
\end{figure*}
The fact that the $\push$ operation cannot decrease the size of a node has a useful corollary:
\begin{proposition}
If $t\vdash^* t'$, and $x$ is a heavy (basket, resp.) node of $t$, then
$x$ is a heavy (basket, resp.) node of $t'$ as well.
Consequently, if $x$ is light in $t'$, then $x$ is light in $t$ as well.
\end{proposition}
\begin{proof}
Retaining heaviness simply comes from $\size(t,x)\leq\size(t',x)$.
When $x$ is a basket node of $t$, say having a heavy child $y$, then by $x\preceq_t y$
and $t\vdash^* t'$ (that is, $t\preceq t'$) we get $x\preceq_{t'} y$ as well, hence $x$ has a (unique) child $z$ which is an ancestor of $y$
in $t'$ (alliwing $y=z$), thus by $\size(t',z)\geq\size(t',y)\geq\size(t,y)$ we get that $z$ is heavy in $t'$ hence $x$ is a basket of $t'$ as well.
\end{proof}

We will frequently sum the sizes of the light children of nodes, so for a tree $t$
and node $x$ of $t$, let $\sumlight(t,x)$ stand for $\sumsize(t,x,H)$, the total
size of the light children of $x$ in $t$.

We introduce a charge function $c$ as follows: for a tree $t$ and node $x\in V_t$, let
the \emph{charge} of $x$ in $t$, denoted $c(t,x)$, be
\[c(t,x)=\begin{cases}0&\hbox{if }x\hbox{ is a light node of }t,\\\sumlight(t,x)-H&\hbox{ otherwise}.\end{cases}\]
It's worth observing that when $x$ is a non-basket heavy node of $t$, then $\sumlight(t,x)=\size(t,x)-1$ (since by being non-basket,
each child of $x$ is light).

Note that in particular the charge of a node is computable from the multiset of the sizes of its children
(or, from the multiset of the sizes of its \emph{light} children along with its own size).
Let $t'=\push(t,x,y)$ and let $z=\parent_t(y)$ be the common parent of $x$ and $y$ in $t$. Then, $c(t,w)=c(t',w)$ for each node $w\notin\{y,z\}$ since for those nodes this multiset does not change.
We now list how $c(t',y)$ and $c(t',z)$ vary depending on whether $y$ and $x$ are light or heavy nodes. Note that $\size(t',z)=\size(t,z)$
(thus the charge of $z$ can differ in $t$ and $t'$ only if $z$ is heavy and $\sumlight(t',z)\neq\sumlight(t,z)$) and
$\size(t',y)=\size(t,x)+\size(t,y)$.
\begin{enumerate}
\item[i)] If $x$ and $y$ are both light nodes with $\size(t,x)+\size(t,y)\leq H$
  (i.e. $y$ remains light in $t'$ after $x$ is pushed below $y$) then
  $\sumlight(t',z)=\sumlight(t,z)$ hence $c(t,z)=c(t',z)$ and $c(t',y)=c(t,y)=0$
  since $y$ is still light in $t'$.
\item[ii)] If $x$ and $y$ are both light nodes with $\size(t,x)+\size(t,y)>H$
  (i.e. $y$ becomes heavy due to the pushing), then $z$ is heavy as well (both
  in $t$ and in $t'$ since $\size(t,z)=\size(t',z)$). Then since $\sumlight(t',z)=\sumlight(t,z)-(\size(t,x)+\size(t,y))$, we get $c(t',z)=c(t,z)-(\size(t,x)+\size(t,y))$
  and $c(t',y)=\size(t,y)+\size(t,x)-(H+1)$ since $y$ is heavy in $t'$ having only light
  children (by the assumption, $x$ is still light and each child of $y$ already present in
  $t$ have to be also light since $y$ itself is light in $t$).
\item[iii)] If $x$ is heavy and $y$ is in light in $t$, then after pushing, $y$ becomes heavy as well (and $z$ is also heavy). Then, by $\sumlight(t',z)=\sumlight(t,z)-\size(t,y)$ (since the child $y$ of $z$ loses its status of being light) we have $c(t',z)=c(t,z)-\size(t,y)$ and in $t'$, $y$ is a heavy node with light children of total size $\size(t,y)-1$, hence $c(t',y)=\size(t,y)-(H+1)$ (while $c(t,y)=0$ since $y$ is light in $t$).
\item[iv)] If $x$ is light and $y$ is heavy in $t$, then $z$ is heavy as well,
 $z$ loses its light child $x$ while $y$ gains the very same light child, hence
 $c(t',z)=c(t,z)-\size(t,x)$ and $c(t',y)=c(t,y)+\size(t,x)$.
\item[v)] Finally, if both $x$ and $y$ are heavy, then $z$ is heavy as well, and
  since neither $z$ nor $y$ loses or gains a light child, $c(t',z)=c(t,z)$ and
  $c(t',y)=c(t,y)$.
\end{enumerate}
Observe that in Cases i), iv) and v) we have $\mathop\sum\limits_{x\in V_t}c(t,x)=\mathop\sum\limits_{x\in V_t}c(t',x)$ while in Cases ii) and iii) it holds that $\mathop\sum\limits_{x\in V_t}c(t,x)>\mathop\sum\limits_{x\in V_t}c(t',x)$ (namely, the total charge decreases by $H+1$ in these two cases).

Now for a flat tree $t_0$ it is easy to check that $\mathop\sum\limits_{x\in V_{t_0}}c(t_0,x)$ is zero:
we have $\sumlight(t_0,\root_{t_0})=(K+1)\cdot H$ by assumption, hence
the root node is in particular heavy and $c(t_0,\root_{t_0})=K\cdot H$,
the empty baskets (having no light node at all) have charge $-H$ each and there are $K$ of them, finally,
all the heavy nodes of size $H+1$ have charge zero as well as each light node, making up
the total sum of charges to be zero.

The charge function we introduced turns out to be particularly useful due to the following fact:
\begin{proposition}
If $t$ is a Union tree, then $c(t,x)\geq 0$ for each node of $t$.
\end{proposition}
\begin{proof}
If $x$ is a light node of $t$, then $c(t,x)=0$ and we are done.

If $x$ is a heavy node which is \emph{not} a basket, then all of its children are light nodes, thus $\sumlight(t,x)=\size(t,x)-1$, making $c(t,x)=\size(t,x)-(H+1)$ which is nonnegative since $\size(t,x)>H$ by $x$ being heavy.

Finally, if $x$ is a basket node of $t$, then it has at least one heavy child.
Let $y$ be a lightest heavy child of $x$. Since $t$ is a Union tree, we have
$\sumsize(t,x,\size(t,y)-1)\geq \size(t,y)-1$. By the choice of $y$, every child of
$x$ which is lighter than $y$ is light itself, thus $\sumsize(t,x,\size(t,y)-1)=\sumsize(t,x,H)=\sumlight(t,x)$,
moreover, $\size(t,y)-1\geq H$ since $y$ is heavy,
thus $\sumlight(t,x)\geq H$ yielding $c(t,x)=\sumlight(t,x)-H\geq 0$ and the statement is proved.
\end{proof}

Thus, since a flat tree has total charge $0$ while in a Union tree each node has nonnegative charge, and
the push operation either decreases the total charge in Cases ii) and iii) above, or leaves the total charge unchanged in the other cases,
we get the following:
\begin{lemma}
Suppose $t_0\vdash^*t$ for a flat tree $t_0$ and a Union tree $t$.
Then for any sequence of push operations transforming $t_0$ into $t$,
each push operation is of type i), iv) or v) above, i.e.
\begin{itemize}
\item either a light node is pushed into a light node, yielding still a light node;
\item or some node (either heavy or light) is pushed into a heavy node.
\end{itemize}
Moreover, the charge of each node of $t$ has to be zero.
\end{lemma}
In particular, a heavy node $x$ has charge $0$ in a tree $t$ if and only if $\sumlight(t,x)=H$ yielding also that each non-basket heavy node of $t$
has to be of size exactly $H+1$.
It's also worth observing that by applying the above three operations one does not \emph{create} a new heavy node: light nodes remain light
all the time.
Moreover, if for a basket node $\sumlight(t,x)=H$ in a Union tree $t$, then there has to exist a heavy child of $x$ in $t$
of size exactly $H+1$ (by $x$ being a basket node, there exists a heavy child and if the size of the lightest heavy child
is more than $H+1$, then the Union condition gets violated).

Recall that in a flat tree $t_0$ with $K$ empty baskets, there are $K+1$ baskets in total (the depth-one empty baskets and the root node),
and there are $K+1$ non-basket heavy nodes (one in each basket, initially), each having size $H+1$, and all the other nodes are light.

Thus if $t_0\vdash^*t$ for some Union tree $t$, then the set of non-basket heavy nodes coincide in $t_0$ and in $t$,
and also in $t$, the size of each such node $x$ is still $H+1$. In particular, one cannot increase the size of $x$
by pushing anything into $x$.

Summing up the result of the above reasoning we have:
\begin{lemma}
Suppose $t_0\vdash^*t$ for some flat tree $t_0$ and some Union tree $t$. Then for any pushing sequence transforming $t_0$ into $t$,
each step $t_i\vdash t_{i+1}$ of the push chain with $t_{i+1}=\push(t_i,x,y)$ has to be one of the following two forms:
\begin{itemize}
\item $x$ and $y$ are light in $t_i$ and in $t_{i+1}$ as well.
\item $y$ is a basket in $t_i$ (and in $t_{i+1}$ as well).
\end{itemize}
\end{lemma}
We can even restrict the order in which the above operations are applied.
\begin{lemma}
\label{lemma-baskets-first}
If $t_0\vdash^*t$ for some flat tree $t_0$ and Union tree $t$, then there is a push sequence
$t_0\vdash t_1\vdash\ldots\vdash t_k=t$ with $t_{i+1}=\push(t_i,x_i,y_i)$ and an index $\ell$ such that for each $i\leq\ell$, $y_i$ is a basket node
and for each $i>\ell$, both $x_i$ and $y_i$ are light nodes with $\size(t_i,x_i)+\size(t_i,y_i)\leq H$.
\end{lemma}
\begin{proof}
Indeed, assume $t_{i+1}=\push(t_i,x,y)$ for the light nodes $x$ and $y$ of $t_i$ with $\size(t_i,x)+\size(t_i,y)\leq H$
and $t_{i+2}=\push(t_{i+1},z,w)$ for the basket $w$ of $t_{i+1}$. Then we can modify the sequence as follows:
\begin{itemize}
\item if $z=y$, then we can get $t_{i+2}$ from $t_i$ by pushing $x$ and $y$ into the basket $w$ first, then pushing $x$ into $y$;
\item otherwise we can simply swap the two push operations since $w$ cannot be either $x$ or $y$ (since light nodes are not baskets),
  nor descendants of $x$ or $y$, thus $z$ and $w$ are already siblings in $t_i$ as well, hence $z$ can be pushed into $w$
  also in $t_i$, and afterwards since $x$ and $y$ are siblings in the resulting tree, $x$ can be pushed into $y$.
\end{itemize}
Applying the above modification finitely many times we arrive to a sequence of trees satisfying the conditions of the Lemma.
Thus if $t_0\vdash^*t$ for some flat tree $t_0$ and Union tree $t$, we can assume that first we push nodes into baskets,
then light nodes into light nodes, yielding light nodes.
\end{proof}
However, it turns out the latter type of pushing cannot fix the Union status of the trees we consider.

\begin{proposition}
\label{prop-light-push-cannot-save-us}
Suppose $t$ is a tree with a basket node $x$ violating the Union condition, i.e. for some child $y$ of $x$
it holds that $\sumsize(t,x,\size(t,y)-1)<\size(t,y)-1$.
Then for any tree $t'=\push(t,z,w)$ with $z$ and $w$ being light nodes with total size at most $H$
we have that $x$ still violates the Union condition in $t'$.
\end{proposition}
\begin{proof}
If $z$ (and $w$) are not children of $x$, then $\children(t,x)=\children(t',x)$ (since in particular $w\neq x$ by $w$ being light and $x$ being a basket)
and each child of $x$ also has the same size in $t'$ as in $t$, hence the Union condition is still violated.
Now assume $z$ and $w$ are children of $x$. Upon pushing, $x$ loses a child of size $\size(t,z)$
and a child of size $\size(t,w)$ and gains a child of size $\size(t,z)+\size(t,w)$. It is clear that
$\sumsize(t,x,W)\geq\sumsize(t',x,W)$ for each possible $W$ then (equality holds when $W\geq\size(t,z)+\size(t,w)$
or $W<\min\{\size(t,w),\size(t,z)\}$ and strict inequality holds in all the other cases).
It is also clear that there is at least one child $y'$ of $x$ in $t'$ such that $\size(t,y)\leq\size(t',y')$:
if $y\neq z$ then $y'=y$ suffices, otherwise $y'=w$ is fine. Now let $y_0$ be such a child with $W=\size(t',y_0)$ being
the minimum possible. Then $\sumsize(t',x,W-1)=\sumsize(t',x,\size(t,y)-1)\leq\sumsize(t,x,\size(t,y)-1)<\size(t,y)-1\leq\size(t',y_0)-1$
and thus $x$ still violates the Union condition in $t'$ as well.
\end{proof}
Hence we have:
\begin{proposition}
\label{prop-save-us-lord-almighty-basket-push-pretty-please}
	Suppose $t_0$ is a flat tree. Then, $t_0$ is a Union-Find tree if and only if $t_0\vdash t_1\vdash\ldots\vdash t_k$ for
	some Union tree $t_k$ such that at each step $t_i\vdash t_{i+1}$ of the chain with $t_{i+1}=\push(t_i,x,y)$,
	the node $y$ is a basket node in $t_i$ (and consequently, in all the other members of the sequence as well).
\end{proposition}
\begin{proof}
Observe that initially in any flat tree $t_0$, only basket nodes violate the Union condition.
Moreover, by pushing arbitrary nodes into baskets only the baskets' Union status can change (namely,
upon pushing into some basket $x$, the status of $x$ and its parent can change, which is also a basket).
Thus, after pushing nodes into baskets we either already have a Union tree, or not, but in the latter case
we cannot transform the tree into a Union tree by pushing light nodes into light nodes by Proposition~\ref{prop-light-push-cannot-save-us}.
\end{proof}
Hence we arrive to the following characterization:
\begin{proposition}
\label{prop-final}
Assume $t_0$ is a flat tree having $K$ empty baskets.
Then, $t_0$ is a Union-Find tree if and only if the set $\mathcal{L}$ of its depth-one light nodes
can be partitioned into sets $L_1,\ldots,L_{K+1}$ such that for each $1\leq i\leq K+1$,
$\sum_{x\in L_i}\size(t_0,x)=H$ and for each $y\in L_i$, $\sum_{z\in L_i,\size(t_0,z)<\size(t_0,y)}\size(t_0,z)\geq \size(t_0,y)-1$.
\end{proposition}
\begin{proof}
Recall that non-basket nodes of a flat tree satisfy the Union condition.
Assume the set $\mathcal{L}$ of the depth-one light nodes of $t_0$ can be partitioned into sets $L_i$, $i=1,\ldots,K+1$ as above.
Let $y_1,\ldots,y_K$ be the empty basket nodes of $t_0$. Then by pushing every member of $L_i$ into the basket $y_i$
(and leaving members of $L_{K+1}$ at depth one) we arrive to a tree $t$ whose basket nodes satisfy the Union condition
(their light children do not violate the condition due to the assumption on $L_i$, and their only heavy child 
having size $H+1$ does not violate the condition either since $\sumlight(t,x)=\sum_{x\in L_i}\size(t_0,x)=H$ due also to
the assumption on $L_i$. Finally, the root node has $K$ children (the initially empty baskets) of size $2H+2$ but since
it has light children of total size $H$ and a heavy child of size $H+1$, their sizes summing up to $2H+1$, the depth-one baskets
also do not violate the Union condition.
Hence $t_0\vdash^*t$ for some Union tree $t$, thus is a Union-Find tree.

For the other direction, assume $t_0$ is a Union-Find tree. Then since $t_0$ is also a flat tree, it can be transformed into a Union tree $t$
by repeatedly pushing nodes into baskets only. Since the pushed nodes' original parents are baskets as well (since they are parents of
the basket into which we push), a node is a child of a basket node in $t$ if and only if it is a child of a (possibly other) basket node in $t_0$.
We also know that the charge of each node in the Union tree we gain at the end has to be zero, in particular, each basket $x$
still has to have a heavy child of size exactly $H+1$ and $\sumlight(t,x)=H$ has to hold. Let $y_1,\ldots,y_{K+1}$ stand for the
basket nodes of $t_0$ (and of $t$ as well): then, defining $L_i$ as
the set of light children of $y_i$ in $t$ suffices.
\end{proof}

At this point we recall that the following problem $3$-\textsc{Partition} is $\NP$-complete in the strong sense~\cite{Garey:1979:CIG:578533}:
given a list $a_1,\ldots,a_{3m}$ of positive integers with the value $B=\frac{\sum_{i=1}^{3m}a_i}{m}$ being an integer,
such that for each $1\leq i\leq 3m$ we have $\frac{B}{4}<a_i<\frac{B}{2}$,
does there exist a partition $\mathcal{B}=\{B_1,\ldots,B_k\}$ of the set $\{1,\ldots,3m\}$ satisfying
$\sum_{i\in B_j}a_i=B$ for each $1\leq j\leq k$?

(Here ``in the strong sense'' means that the problem remains $\NP$-complete even if the numbers are encoded in unary.)

Observe that by the condition $\frac{B}{4}<a_i<\frac{B}{2}$ each set $B_j$ has to contain exactly three elements
(the sum of any two of the input numbers is less than $B$ and the sum of any four of them is larger than $B$),
thus in particular in any solution above $k=m$ holds.

Note also that for any given instance $I=a_1,\ldots,a_{3m}$ of the above problem and any given offset $c\geq 0$, the instance
$I'=a'_1,\ldots,a'_{3m}$ with $a'_i=a_i+c$ for each $i$ and with $B'=B+3c$, the set of solutions of $I$ and $I'$ coincide.
Indeed, the instance $I'$ still satisfies
\begin{eqnarray*}
\frac{B'}{4}&=&\frac{B+3c}{4}=\frac{B}{4}+\frac{3}{4}c<\frac{B}{4}+c<a_i+c=a'_i\\
&=&a_i+c<\frac{B}{2}+\frac{3}{2}c=\frac{B+3c}{2}=\frac{B'}{2},
\end{eqnarray*}
hence each solution $\mathcal{B}'=\{B'_1,\ldots,B'_k\}$ of $I'$ still contains triplets and $\sum_{i\in B'_j}a'_i=3c+\sum_{i\in B'_j}a_i$
which is $B'$ if and only if $\sum_{i\in B'_j}a_i=B$, thus any solution of $I'$ is also a solution of $I$, the other direction being also straightforward to check.

Thus, by setting the above offset to $c=\lceil(1+\max\{2^{\lceil \log B\rceil},2^4\}-B)/3\rceil$ (in which case $B'=2^D+d$ for some suitable integers $D>3$ and
$d\in\{1,2,3\}$) we get that the following problem is also strongly $\NP$-complete:
\begin{definition}[$3$-\textsc{Partition}']
{\noindent\bf Input}: A sequence $a_1,\ldots,a_{3m}$ of positive integers such that $B=\frac{\sum_{i=1}^{3m}a_i}{m}$ is a positive integer of
the form $2^D+d$ for some integer $D>3$ and $d\in\{1,2,3\}$
and $\frac{B}{4}<a_i<\frac{B}{2}$ for each $i=1,\ldots,3m$.

{\noindent\bf Output}: Does there exist a partition $\mathcal{B}=\{B_1,\ldots,B_k\}$ of the set $\{1,\ldots,3m\}$ satisfying
$\sum_{i\in B_j}a_i=B$ for each $j\in\{1,\ldots,k\}$?

Observe that in any solution, $k=m$ and each $B_j$ consists of exactly three elements.
\end{definition}

We are now ready to show the $\NP$-completeness of the recognition of Union-Find trees via a (logspace, many-to-one) reduction
from the $3$-\textsc{Partition}' problem to it. To each instance $I=a_1,\ldots,a_{3m}$ of the $3$-\textsc{Partition}' problem
we associate the following flat tree $t(I)$:
\begin{itemize}
\item The number $K$ of empty baskets in $t(I)$ is $m-1$.
\item The heaviness threshold parameter $H$ of the tree is $B+2^{D-1}-1$ where $B=2^D+d$ is the target sum $\frac{\sum_{i=1}^{3m}a_i}{m}$
with $D>3$ being an integer and $d\in\{1,2,3\}$.
\item There are $3m+(D-1)m$ light nodes of depth one in $t(I)$: first, to each member $a_i$ of the input a light node $x_i$ of
  size $a_i$ is associated, and to each index $1\leq i\leq m$ and $0\leq j< D-1$, a light node $y_i^j$ of size
  $w_i^j=2^j$ is associated.
\end{itemize}
Note that since the $3$-\textsc{Partition}' problem is strongly $\NP$-complete we can assume that the input numbers $a_i$ are
encoded in unary, thus the tree $t(I)$ can be indeed built in logspace (and in polytime).

(See Figure~\ref{fig-reduction} for an example.)

\begin{figure*}[h]
\centering
\begin{tikzpicture}[every node/.style={thick,circle,inner sep=0pt,minimum size=0.1cm}]
\node (root) at (3.5,0) [draw] {};
\drawweight{-3.6}{-0.5}{4}{0.5}{-1}{(root)}
\drawweight{-2.8}{-0.5}{4}{0.5}{-1}{(root)}
\drawweight{-2}{-0.5}{4}{0.5}{-1}{(root)}
\drawweight{-1.2}{-0.5}{4}{0.5}{-1}{(root)}
\drawweight{-0.4}{-0.5}{4}{0.5}{-1}{(root)}
\drawweight{0.4}{-0.5}{5}{0.6}{-1}{(root)}
\drawweight{1.2}{-0.5}{5}{0.6}{-1}{(root)}
\drawweight{2.1}{-0.5}{6}{0.7}{-1}{(root)}
\drawweight{3}{-0.5}{6}{0.7}{-1}{(root)}
{\node (r) at (3.5,-0.5) [draw] {}; \path (r) edge (root); }
{\node (r) at (3.7,-0.5) [draw] {}; \path (r) edge (root); }
{\node (r) at (3.9,-0.5) [draw] {}; \path (r) edge (root); }
{\node (r) at (4.1,-0.5) [draw] {}; \path (r) edge (root); \node (c) at (4.1,-1) [draw] {}; \path (r) edge (c);}
{\node (r) at (4.3,-0.5) [draw] {}; \path (r) edge (root); \node (c) at (4.3,-1) [draw] {}; \path (r) edge (c);}
{\node (r) at (4.5,-0.5) [draw] {}; \path (r) edge (root); \node (c) at (4.5,-1) [draw] {}; \path (r) edge (c);}
\drawweight{5}{-0.5}{3}{0.5}{-1}{(root)}
\drawweight{5.7}{-0.5}{3}{0.5}{-1}{(root)}
\drawweight{6.4}{-0.5}{3}{0.5}{-1}{(root)}
\drawweight{8}{-0.5}{24}{2}{-1}{(root)}
{\node (r1) at (9.3,-0.5) [draw] {}; \path (r1) edge (root); \drawweight{9.3}{-1}{24}{2}{-1.5}{(r1)} }
{\node (r1) at (11.5,-0.5) [draw] {}; \path (r1) edge (root); \drawweight{11.5}{-1}{24}{2}{-1.5}{(r1)} }
\end{tikzpicture}
\caption{Illustrating the reduction: $t(I)$ for $I=(5,5,5,5,5,6,6,7,7)$, $m=3$, $B=17$. Then $D=4$, $d=1$ and $H=24$. Small weights $w_i^j$ are of size $1,2,4$.}
\label{fig-reduction}
\end{figure*}
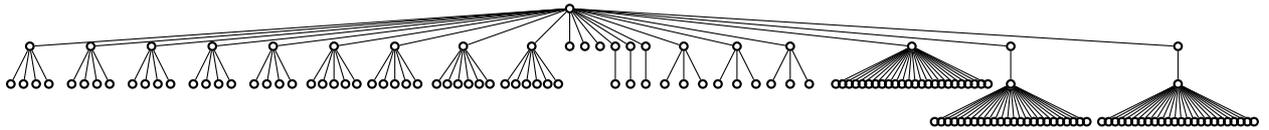

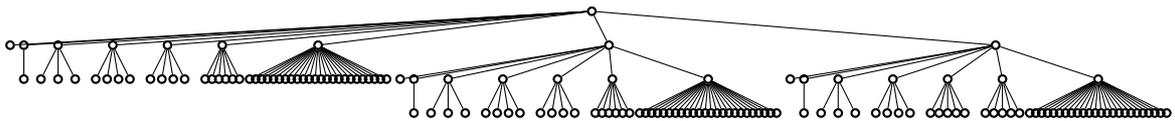
\begin{figure*}[h]
\centering
\begin{tikzpicture}[scale=0.9,every node/.style={thick,circle,inner sep=0pt,minimum size=0.1cm}]
\node (root) at (8.5,0.5) [draw] {};
{\node (r1) at (0,0) [draw] {}; \path (r1) edge (root);}
{\node (r) at (0.2,0) [draw] {}; \path (r) edge (root); \node (c) at (0.2,-0.5) [draw] {}; \path (r) edge (c);}
\drawweight{0.7}{0}{3}{0.5}{-0.5}{(root)}
\drawweight{1.5}{0}{4}{0.5}{-0.5}{(root)}
\drawweight{2.3}{0}{4}{0.5}{-0.5}{(root)}
\drawweight{3.1}{0}{6}{0.5}{-0.5}{(root)}
\drawweight{4.5}{0}{24}{2}{-0.5}{(root)}
{
\node(r1) at (8.75,0) [draw] {}; \path (r1) edge (root);
{\node (r) at (5.7,-0.5) [draw] {}; \path (r) edge (r1);}
{\node (r) at (5.9,-0.5) [draw] {}; \path (r) edge (r1); \node (c) at (5.9,-1) [draw] {}; \path (r) edge (c);}
\drawweight{6.4}{-0.5}{3}{0.5}{-1}{(r1)}
\drawweight{7.2}{-0.5}{4}{0.5}{-1}{(r1)}
\drawweight{8.0}{-0.5}{4}{0.5}{-1}{(r1)}
\drawweight{8.8}{-0.5}{6}{0.5}{-1}{(r1)}
\drawweight{10.2}{-0.5}{24}{2}{-1}{(r1)}
}
{
\node(r1) at (14.4,0) [draw] {}; \path (r1) edge (root);
{\node (r) at (11.4,-0.5) [draw] {}; \path (r) edge (r1);}
{\node (r) at (11.6,-0.5) [draw] {}; \path (r) edge (r1); \node (c) at (11.6,-1) [draw] {}; \path (r) edge (c);}
\drawweight{12.1}{-0.5}{3}{0.5}{-1}{(r1)}
\drawweight{12.9}{-0.5}{4}{0.5}{-1}{(r1)}
\drawweight{13.7}{-0.5}{5}{0.5}{-1}{(r1)}
\drawweight{14.5}{-0.5}{5}{0.5}{-1}{(r1)}
\drawweight{15.9}{-0.5}{24}{2}{-1}{(r1)}
}
\end{tikzpicture}
\caption{The Union tree corresponding to a solution of the instance on Figure~\ref{fig-reduction}.}
\label{fig-solution}
\end{figure*}

In order to ease notation we say that a member $a_i$ of some sequence (or multiset) $a_1,\ldots,a_n$
satisfies the Union condition if the sum of the members of the sequence that are less than $a_i$ is at least $a_i-1$,
i.e. $\mathop\sum\limits_{a_j<a_i}a_j\geq a_i-1$. In addition we say that the sequence itself satisfies
the Union condition if each of its elements does so.
It is clear that a node $x$ of a tree $t$ satisfies the Union condition if and only if the multiset $\{\size(t,y):y\in\children(t,x)\}$
does so.

The following lemma states that the above construction is indeed a reduction:
\begin{lemma}
For any instance $I$ of the $3$-\textsc{Partition}' problem, $I$ has a solution iff $t(I)$ is a Union-Find tree.
\end{lemma}
\begin{proof}
For one direction, assume $\{B_1,\ldots,B_m\}$ is a solution of the instance $I=a_1,\ldots,a_{3m}$.
Then the multiset of the sizes of the light nodes can be partitioned into sets $\mathcal{L}=(L_1,\ldots,L_m)$
as $L_i=\{a_j:j\in B_i\}\cup\{w_i^j:0\leq j<D-1\}$. It is clear that $\sum_{\ell\in L_i}\ell=B+2^{D-1}-1=H$.
Thus, by Proposition~\ref{prop-final} we only have to show that each $L_i$ satisfies the Union condition.
For the elements $w_i^j$ of size $2^j$ it is clear since $\sum_{j'=0}^{j-1}w_i^{j'}=2^j-1$.
Now let $a$ be the smallest element of $B_i$. Then, since $B_i$ consists of three integers summing up to $B$,
we get that $a\leq\frac{B}{3}$. On the other hand, $\frac{B}{4}<a$ by the definition of the $3$-\textsc{Partition}' problem.
Recall that $B=2^D+d$ for some $d\in\{1,2,3\}$: we get that $2^{D-2}<a$, thus in particular each weight $w_i^j$ is smaller than $a$.
Summing up all these weights we get $\sum_{j=0}^{D-2}w_i^j=2^{D-1}-1$.
We claim that $a\leq\frac{B}{3}=\frac{2^D+d}{3}<2^{D-1}-1$.
Indeed, multiplying by $3$ we get $2^D+d<2^D+2^{D-1}-3$, subtracting $2^D$ and adding $3$ yields $d+3<2^{D-1}$
which holds since $d\leq 3$ and $D>3$.
Thus we have that $a$ satisfies the Union condition as well.
Then, if $b>a$ is also a member of $B_i$, then it suffices to show $a+\sum_{j=0}^{D-2}w_i^j\geq b-1$ but since
$a\geq\frac{B}{4}>2^{D-2}$ and $\sum_{j=0}^{D-2}w_i^j=2^{D-1}-1$, we get the left sum exceeds $2^{D-1}+2^{D-2}-1$
which is at least $b-1$ since $b<\frac{B}{2}=\frac{2^D+d}{2}\leq 2^{D-1}+2^{D-2}$ (since $d\leq 3$ and $2^{D-2}\geq 2$ if $D\geq 3$).
Thus by Proposition~\ref{prop-final}, $t(I)$ is indeed a Union-Find tree. (See Figure~\ref{fig-solution}.)

For the other direction, suppose $t(I)$ is a Union-Find tree. Then by Proposition~\ref{prop-final},
the multiset of the sizes of the light nodes can be partitioned into sets $\mathcal{L}=\{L_1,\ldots,L_m\}$ such 
that each $L_i$ sums up exactly to $H=B+2^{D-1}-1$ and each $L_i$ satisfies the Union condition.

First we show that each $L_i$ contains exactly one ``small'' weight $w_k^j$ for each $j=0,\ldots,D-2$.
(Note that each $a_j$ exceeds $2^{D-2}$ hence the name of these weights.)
We prove this by induction on $j$. The claim holds for $j=0$ since in a Union-Find tree each
inner node has a leaf child. By induction,
the smallest $j$ members of $L_i$ have sizes $2^0,2^1,\ldots,2^{j-1}$, summing up to $2^j-1$.
Since all the weights $a_k$ are larger than $2^{D-2}$ as well as the weights $w_k^{j'}$ for $j'>j$,
none of these can be the $(j+1)$th smallest integer in $L_i$ without violating the Union condition.
Thus, each set $L_i$ has to have some $w_{k_i}^j$ as its $j+1$th smallest element, and since both the
number of these sets and the number of the small weights of size $2^j$ are $m$ we get that
each set $L_i$ contains exactly one small weight of size $2^j$.

Thus, the small weights sum up exactly to $2^{D-1}-1$ in each of the sets $L_i$, hence the weights $a_i$
have to sum up exactly to $H-(2^{D-1}-1)=2^D+d=B$ in each of these sets, yielding a solution to the
instance $(a_1,\ldots,a_{3m})$ of the $3$-\textsc{Partition}' problem and concluding the proof.
\end{proof}

Thus, since the strongly $\NP$-complete $3$-\textsc{Partition}' problem reduces to the problem of deciding whether a \emph{flat} tree is a Union-Find tree,
and moreover, if a flat tree is a Union-Find tree, then it can be constructed from a Union tree by applying finitely many
path compressions (all one has to do is to ``move out'' the light nodes from the baskets by calling a $\mathrm{find}$ operation
on each of them successively), we have proven the following:
\begin{theorem}
\label{thm-main}
It is $\NP$-complete already for flat trees to decide whether
\begin{itemize}
\item[i)] a given (flat) tree is a Union-Find tree and
\item[ii)] whether a given (flat) tree can be constructed from a Union tree by applying a number of path compression operations.
\end{itemize}
\end{theorem}

\section{Conclusion, future directions}
We have shown that unless $\mathbf{P}=\NP$, there is no efficient algorithm to check whether a given tree is a valid Union-Find tree,
assuming union-by-size strategy and usage of path compression, since the problem is $\NP$-complete. A natural question is whether
the recognition problem remains $\NP$-hard assuming union-by-rank strategy (and of course path compression)? 
Since the heights of merged trees do not add up, only increase by at most one, there is no ``obvious'' way to encode arithmetic
into the construction of these trees, and even the characterization by the push operation is unclear to hold in that case
(since in that setting, path compression can alter the order of the subsequent merging).

\bibliography{biblio}{}

\begin{thebibliography}{10}

\bibitem{Bannai2003208}
H.~Bannai, S.~Inenaga, A.~Shinohara, and M.~Takeda.
\newblock Inferring strings from graphs and arrays.
\newblock {\em Lecture Notes in Computer Science (including subseries Lecture
  Notes in Artificial Intelligence and Lecture Notes in Bioinformatics)},
  2747:208--217, 2003.
\newblock cited By 15.

\bibitem{DBLP:journals/ipl/Cai93}
Leizhen Cai.
\newblock The recognition of union trees.
\newblock {\em Inf. Process. Lett.}, 45(6):279--283, 1993.

\bibitem{DBLP:conf/stacs/ClementCR09}
Julien Cl{\'{e}}ment, Maxime Crochemore, and Giuseppina Rindone.
\newblock Reverse engineering prefix tables.
\newblock In Susanne Albers and Jean{-}Yves Marion, editors, {\em 26th
  International Symposium on Theoretical Aspects of Computer Science, {STACS}
  2009, February 26-28, 2009, Freiburg, Germany, Proceedings}, volume~3 of {\em
  LIPIcs}, pages 289--300. Schloss Dagstuhl - Leibniz-Zentrum fuer Informatik,
  Germany, 2009.

\bibitem{Cormen:2001:IA:580470}
Thomas~H. Cormen, Clifford Stein, Ronald~L. Rivest, and Charles~E. Leiserson.
\newblock {\em Introduction to Algorithms}.
\newblock McGraw-Hill Higher Education, 2nd edition, 2001.

\bibitem{Crochemore2010251}
M.~Crochemore, C.S. Iliopoulos, S.P. Pissis, and G.~Tischler.
\newblock Cover array string reconstruction.
\newblock {\em Lecture Notes in Computer Science (including subseries Lecture
  Notes in Artificial Intelligence and Lecture Notes in Bioinformatics)}, 6129
  LNCS:251--259, 2010.
\newblock cited By 8.

\bibitem{Duval2009281}
J.-P. Duval, T.~Lecroq, and A.~Lefebvre.
\newblock Efficient validation and construction of border arrays and validation
  of string matching automata.
\newblock {\em RAIRO - Theoretical Informatics and Applications},
  43(2):281--297, 2009.
\newblock cited By 12.

\bibitem{Duval2002249}
J.-P. Duval and A.~Lefebvre.
\newblock Words over an ordered alphabet and suffix permutations.
\newblock {\em Theoretical Informatics and Applications}, 36(3):249--259, 2002.
\newblock cited By 19.

\bibitem{Duval:2005:BAB:1131983.1131987}
Jean-Pierre Duval, Thierry Lecroq, and Arnaud Lefebvre.
\newblock Border array on bounded alphabet.
\newblock {\em J. Autom. Lang. Comb.}, 10(1):51--60, January 2005.

\bibitem{Fredman:1989:CPC:73007.73040}
M.~Fredman and M.~Saks.
\newblock The cell probe complexity of dynamic data structures.
\newblock In {\em Proceedings of the Twenty-first Annual ACM Symposium on
  Theory of Computing}, STOC '89, pages 345--354, New York, NY, USA, 1989. ACM.

\bibitem{Galler:1964:IEA:364099.364331}
Bernard~A. Galler and Michael~J. Fisher.
\newblock An improved equivalence algorithm.
\newblock {\em Commun. ACM}, 7(5):301--303, May 1964.

\bibitem{Garey:1979:CIG:578533}
Michael~R. Garey and David~S. Johnson.
\newblock {\em Computers and Intractability: A Guide to the Theory of
  NP-Completeness}.
\newblock W. H. Freeman \& Co., New York, NY, USA, 1979.

\bibitem{Gawrychowski2014337}
P.~Gawrychowski, A.~Jez, and Ł. Jez.
\newblock Validating the knuth-morris-pratt failure function, fast and online.
\newblock {\em Theory of Computing Systems}, 54(2):337--372, 2014.
\newblock cited By 7.

\bibitem{MR2544434}
Tomohiro I, Shunsuke Inenaga, Hideo Bannai, and Masayuki Takeda.
\newblock Counting parameterized border arrays for a binary alphabet.
\newblock In {\em Language and automata theory and applications}, volume 5457
  of {\em Lecture Notes in Comput. Sci.}, pages 422--433. Springer, Berlin,
  2009.

\bibitem{I20116959}
Tomohiro I, Shunsuke Inenaga, Hideo Bannai, and Masayuki Takeda.
\newblock Verifying and enumerating parameterized border arrays.
\newblock {\em Theoretical Computer Science}, 412(50):6959 -- 6981, 2011.

\bibitem{MR2894365}
Tomohiro I, Shunsuke Inenaga, Hideo Bannai, and Masayuki Takeda.
\newblock Verifying and enumerating parameterized border arrays.
\newblock {\em Theoret. Comput. Sci.}, 412(50):6959--6981, 2011.

\bibitem{I2014316}
Tomohiro I, Shunsuke Inenaga, Hideo Bannai, and Masayuki Takeda.
\newblock Inferring strings from suffix trees and links on a binary alphabet.
\newblock {\em Discrete Applied Mathematics}, 163, Part 3:316 -- 325, 2014.
\newblock Stringology Algorithms.

\bibitem{Knight:1989:UMS:62029.62030}
Kevin Knight.
\newblock Unification: A multidisciplinary survey.
\newblock {\em ACM Comput. Surv.}, 21(1):93--124, March 1989.

\bibitem{Kucherov2013915}
Gregory Kucherov, Lilla Tóthmérész, and Stéphane Vialette.
\newblock On the combinatorics of suffix arrays.
\newblock {\em Information Processing Letters}, 113(22–24):915 -- 920, 2013.

\bibitem{Lu2002}
Weilin Lu, P.~J. Ryan, W.~F. Smyth, Yu~Sun, and Lu~Yang.
\newblock Verifying a border array in linear time.
\newblock {\em J. Comb. Math. Comb. Comput}, 42:223--236, 2000.

\bibitem{DBLP:journals/csr/McConnellMNS11}
Ross~M. McConnell, Kurt Mehlhorn, Stefan N{\"{a}}her, and Pascal Schweitzer.
\newblock Certifying algorithms.
\newblock {\em Computer Science Review}, 5(2):119--161, 2011.

\bibitem{Starikovskaya201514}
Tatiana Starikovskaya and Hjalte~Wedel Vildh\o{}j.
\newblock A suffix tree or not a suffix tree?
\newblock {\em Journal of Discrete Algorithms}, 32:14 -- 23, 2015.
\newblock StringMasters 2012 2013 Special Issue (Volume 2).

\bibitem{Tarjan:1975:EGB:321879.321884}
Robert~Endre Tarjan.
\newblock Efficiency of a good but not linear set union algorithm.
\newblock {\em J. ACM}, 22(2):215--225, April 1975.

\end{thebibliography}
\bibliographystyle{plain}

\end{document}